\documentclass[journal,onecolumn,12pt]{IEEEtran}

\usepackage{times,amsmath,amsthm,epsfig,amssymb,graphicx,amsfonts}
\usepackage{psfrag}
\usepackage{ifpdf}
\usepackage{cite}
\usepackage{array}
\usepackage{multirow}
\usepackage{graphicx}
\usepackage[usenames,dvipsnames]{color}
\usepackage{dsfont}
\usepackage{tabu}
\usepackage{subfigure}
\usepackage{epsfig}
\usepackage{citesort}
\usepackage{setspace}
\usepackage{latexsym}
\usepackage{amssymb}
\usepackage{amsmath}

\usepackage{mathtools}

\title{On the DoF of Two-way $2\times2\times2$ Relay Networks\\with or without Relay Caching}

\author{Mehdi Ashraphijuo, Vaneet Aggarwal, and Xiaodong Wang \thanks{M. Ashraphijuo and X. Wang are with the Electrical Engineering Department, Columbia University, New York, NY 10027 (e-mail: mehdi@ee.columbia.edu, wangx@ee.columbia.edu). V. Aggarwal is with the School of Industrial Engineering, Purdue University, West Lafayette, IN 47907 (e-mail: vaneet@purdue.edu).}}

\newtheorem{theorem}{Theorem}

\newtheorem{proposition}{Proposition}
\newtheorem{lemma}{Lemma}
\newtheorem{corollary}{Corollary}
\newtheorem{conjecture}{Conjecture}
\newtheorem{remark}{Remark}

\begin{document}
\maketitle

\begin{abstract}
Two-way relay is potentially an effective approach to spectrum sharing and aggregation by allowing simultaneous bidirectional transmissions between source-destinations pairs. This paper studies the two-way $2\times2\times2$ relay network, a class of four-unicast networks, where there are four source/destination nodes and two relay nodes, with each source sending a message to its destination. We show that without relay caching the total degrees of freedom (DoF) is bounded from above by $8/3$, indicating that bidirectional links do not double the DoF (It is known that the total DoF of one--way $2\times2\times2$ relay network is $2$.). Further, we show that the DoF of $8/3$ is achievable for the two-way $2\times2\times2$ relay network with relay caching. Finally, even though the DoF of this network is no more than $8/3$ for generic channel gains, DoF of $4$ can be achieved for a symmetric configuration of channel gains.
\end{abstract}

{\bf Index terms:} Degrees of freedom (DoF) , four-unicast channels, two-way relay network, caching.


\newpage

\section{Introduction}

To meet the demand of future mobile networks, it is important to understand the fundamental capacity of networks. While characterizing network capacity is in general unsolved, there has been considerable progress in two research fronts. The first one focuses on single-flow multi-hop networks, in which one source aims to send the same message to one or more destinations, using multiple relay nodes. Since in this scenario all destination nodes are interested in the same message, there is effectively only one information stream in the network. Starting from the max-flow-min-cut theorem of Ford-Fulkerson \cite{ford1956maximal}, there has been significant progress on this problem \cite{avestimehr2011wireless}.

The second research direction focuses on multi-flow wireless networks with only one-hop between the sources and the destinations, i.e., the interference channel. While the capacity of the interference channel remains unknown, there has been a variety of approximate capacity results, such as constant-gap capacity approximations \cite{etkin2008gaussian,bresler2010approximate,bresler2008two} and degrees of freedom (DoF) characterizations \cite{cadambe2008interference,etkin2006degrees,jafar2008degrees,motahari2014real,maddah2008communication,jafar2007degrees}.

The two-way communication between two users was first studied by Shannon \cite{Shannon}. Recently, there have been many attempts to demonstrate two-way communications experimentally \cite{Chen:1998aa,Khandani,Bliss:2007,Radunovic:2009aa,Aryafar,tech_report,Bharadia}. The two-way relay channel where two users communicate to each other in the presence of relays, has been widely studied \cite{b1,b2,i4,i5,nam2010capacity,i7,i8,Rankov,s2,s5,s8,s16,s20,s22,Avestim-t2,s4,s12,s13,JJ}. Two-unicast channels consist of two sources and two destinations communicating through a general network. The DoF for one-way $2\times2\times2$ fully-connected two-unicast channels has been studied in \cite{gou2012aligned}, and further extended with interfering relays in  \cite{gou2011aligned}. These results were further generalized to one-way $2\times2\times2$ non-layered topology in \cite{gou2014toward,gou2011degrees2}. General one-way two-unicast channel has been considered in \cite{wang2011multiple,shomorony2013two} and it was shown in \cite{shomorony2013two} that the DoF for any topology takes one of the values in $\{1,\frac{3}{2},2\}$, depending on the topology. Two-way two-unicast  channels have been studied for a single relay in \cite{wang2013degrees,wang2014beyond,xin2011coordinated}. In \cite{lee2013achievable}, three different achievability strategies for two-way MIMO $2\times2\times2$ fully-connected channel are proposed. A finite-field two-way two-unicast model is also studied in \cite{maier2013cyclic,hong2013two}.

Caching is a technique to reduce traffic load by exploiting the high degree of asynchronous content reuse and the fact that storage is cheap and ubiquitous in today's wireless devices \cite{2cc,3cc}. During off-peak periods when network resources are abundant, some content can be stored at the wireless edge (e.g., access points or end user devices), so that demands can be met with reduced access latencies and bandwidth requirements. The caching problem has a long history, dating back to the work by Belady in 1966 \cite{1ma}. There are various forms of caching, i.e., to store data at user ends, relays, etc. \cite{4cc}. However, using the uncoded data on devices can result in an inefficient use of the aggregate cache capacity \cite{5cc}. The caching problem consists of a placement phase which is performed offline and an online delivery phase. One important aspect of this problem is the design of the placement phase in order to facilitate the delivery phase. There are several recent works that consider communication scenarios where user nodes have pre-cached information from a fixed library of possible files during the offline phase, in order to minimize the transmission from source during the delivery phase \cite{maddah2014fundamental,ji2015throughput}. There are only a limited number of works on the DoF with caching. In particular, \cite{han2015degrees,han2015improving} study the DoF for the relay and interference channels with caching, respectively, under some assumptions and provide asymptotic results on the DoF as a function of the output of some optimization problems.

In this paper, we study the two-way $2\times2\times2$ relay network, a class of four-unicast networks, also known as the two-way layered interference channel in the literature. We consider a general two-way $2\times2\times2$ relay network where all channel gains are chosen from the same continuous distribution. Even though the one-way $2\times2\times2$ relay network has $2$ DoF, we show that the two-way $2\times2\times2$ relay network has DoF less than or equal to $8/3$. Thus, the bidirectional links cannot double the DoF. In the analysis of cached $2\times2\times2$ relay network, we show the equivalence of our model to the compound MISO broadcast channel and use the existing results on the latter to obtain the achievable DoF of the former. Note that this is the first work relating $2\times2\times2$ relay network and compound MISO broadcast channel.

We further propose a caching strategy in relays for the two-way $2\times2\times2$ relay network based on prefetching uncoded raw bits and delivering linearly encoded messages to facilitate the transmission from relays to destinations. We show that with relay caching, the DoF of $8/3$ is achievable.

Finally for a special case of two-way $2\times2\times2$ relay network where the channels exhibit symmetries, we show that the DoF is $4$. This special case is interesting because (i) This shows that the $2\times2\times2$ topology allows $4$ DoF for some symmetric channel gains while the DoF is outer bounded by $8/3$ for generic channel gains. (ii) The non-invertibility of the symmetric channel matrix plays an important role in achieving DoF=4, which does not hold for the general channel matrix.  (iii) The symmetric channel  is a common model for many results on interference channels \cite{NRx1,NRx2,NRx3,NRx4}; and these results can be obtained only with such symmetric assumptions and the problems remain open otherwise.

The remainder of this paper is organized as follows. In Section \ref{sec2aa}, we give the model for the two-way $2\times2\times2$ relay network. In Section \ref{thm_2_UB_Sec}, the main results on the DoF of the $2\times2\times2$ relay network without relay caching is studied. The DoF of a symmetric two-way $2\times2\times2$ relay network is also investigated in this section. In Section \ref{ofo2}, the results on the DoF of the $2\times2\times2$ relay network with relay caching is presented. Finally, Section \ref{sec5} concludes this paper.


\section{Channel Model and Related Works}\label{sec2aa}

In this section, we first present our system model and then we discuss some recent related results in the literature.

\subsection{Channel Model}\label{vmnbx}
As shown in Fig. \ref{fig:1}, the two-way $2\times2\times2$ relay network consists of four transmitters $S_1,\dots,S_4$, two relays $R_1,R_2$, and four receivers $D_1,\dots,D_4$. Each transmitter $S_i$ has one message that is intended for its receiver $D_i$. Fig. \ref{fig:2} shows the two hops of this system separately. In the first hop (Fig. \ref{fig:2.1}), the signal received at relay $R_k$, $k\in\{1,2\}$ in time slot $m$ is
\begin{equation}\label{wyek1}
Y_{R_k}[m]=\sum_{i=1}^{4}H_{i,R_k}[m]X_i[m]+Z_{R_k}[m],
\end{equation}
where $H_{i,R_k}[m]$ is the channel coefficient from transmitter $S_i$ to relay $R_k$, $X_i[m]$ is the signal transmitted from $S_i$, $Y_{R_k}[m]$ is the signal received at relay $R_k$ and $Z_{R_k}[m]$ is the i.i.d. circularly symmetric complex Gaussian noise with zero mean and unit variance, $i\in\{1,2,3,4\}$, $k\in\{1,2\}$. In the second hop (Fig. \ref{fig:2.2}), the signal received at receiver $D_i$ in time slot $m$ is given by
\begin{equation}
Y_i[m]=\sum_{k=1}^{2}H_{R_k,i}[m]X_{R_k}[m]+Z_i[m],
\end{equation}
where $H_{R_k,i}[m]$ is the channel coefficient from relay $R_k$ to receiver $D_i$, $X_{R_k}[m]$ is the signal transmitted from $R_k$, $Y_i[m]$ is the signal received at receiver $D_i$ and $Z_i[m]$ is the i.i.d. circularly symmetric complex Gaussian noise with zero mean and unit variance, $i\in\{1,2,3,4\}$, $k\in\{1,2\}$. We assume that the channel coefficient values are drawn i.i.d. from a continuous distribution and they are bounded from below and above, i.e., $H_{\min} < | H_{i,R_k}[m] | < H_{\max}$ and $H_{\min} < | H_{R_k,i}[m] | < H_{\max}$ as in \cite{cadambe2008interference}. The relays are assumed to be full-duplex and equipped with caches. Furthermore, the relays are assumed to be causal, which means that the signals transmitted from the relays depend only on the signals received in the past and not on the current received signals and can be described as
\begin{equation}
X_{R_k}[m] = f(Y_{R_k}^{m-1},X_{R_k}^{m-1},C_{R_k}),
\end{equation}
where $X_{R_k}^{m-1}\triangleq(X_{R_k}[1],\dots,X_{R_k}[m-1])$, $Y_{R_k}^{m-1}\triangleq(Y_{R_k}[1],\dots,Y_{R_k}[m-1])$, and $C_{R_k}$ is the cached information in relay ${R_k}$. We assume that source $S_i$, $i\in\{1,2,3,4\}$ knows only channels $H_{i,R_k}$, $k\in\{1,2\}$; relay $R_k$, $k\in\{1,2\}$ knows channels $H_{i,R_k}$,  $H_{R_1,i}$ and $H_{R_2,i}$, $i\in\{1,2,3,4\}$; and destination $D_i$, $i\in\{1,2,3,4\}$ knows only channels $H_{R_k,i}$, $k\in\{1,2\}$.

The source $S_i$, $i\in\{1,2,3,4\}$ has a message $W_i$ that is intended for destination $D_i$. $|W_i|$ denotes the size of the message $W_i$. The rates ${\mathcal R}_i=\frac{\log |W_i|}{n}$, $i\in\{1,2,3,4\}$ are achievable during $n$ channel uses by choosing $n$ large enough, if  the probability of error can be arbitrarily small for all four messages simultaneously. The capacity region ${\mathcal C}=\{({\mathcal R}_1,{\mathcal R}_2,{\mathcal R}_3,{\mathcal R}_4)|({\mathcal R}_1,{\mathcal R}_2,{\mathcal R}_3,{\mathcal R}_4)\in{\mathcal C}\}$ represents the set of all achievable quadruples. The sum-capacity is the maximum sum-rate that is achievable, i.e., ${\mathcal C}_{\Sigma}(P)=\sum_{i=1}^{4}{\mathcal R}^c_i$ where $({\mathcal R}_1^c,{\mathcal R}_2^c,{\mathcal R}_3^c,{\mathcal R}_4^c) = \arg\max _{({\mathcal R}_1,{\mathcal R}_2,{\mathcal R}_3,{\mathcal R}_4)\in{\mathcal C}}\sum_{i=1}^{4} {\mathcal R}_i $ and $P$ is the transmit power at each source node. The DoF is defined as
\begin{equation}
DoF \triangleq \lim_{P\rightarrow\infty}\frac{{\mathcal C}_{\Sigma}(P)}{\log P}=
\sum_{i=1}^{4}\lim_{P\rightarrow\infty}\frac{{\mathcal R}^c_i}{\log P}=\sum_{i=1}^{4}d_i,
\end{equation}
where $d_i \triangleq \lim_{P\rightarrow\infty}\frac{{\mathcal R}^c_i}{\log P}$ is the ${\text{DoF}}$ of source $S_i$, for $i\in\{1,2,3,4\}$. We  assume that channel gains are i.i.d., chosen from the same continuous distribution, and thus the ${\text{DoF}}$ is the result for almost every channel realization (in other words, with probability 1 over the channel realizations). We denote ${\text{DoF}}_C$ as the DoF for the case of with relay caching, ${\text{DoF}}_{NC}$ as the DoF for the case of no relay caching.

\begin{figure}[htbp]
\centering
	\includegraphics[width=9cm]{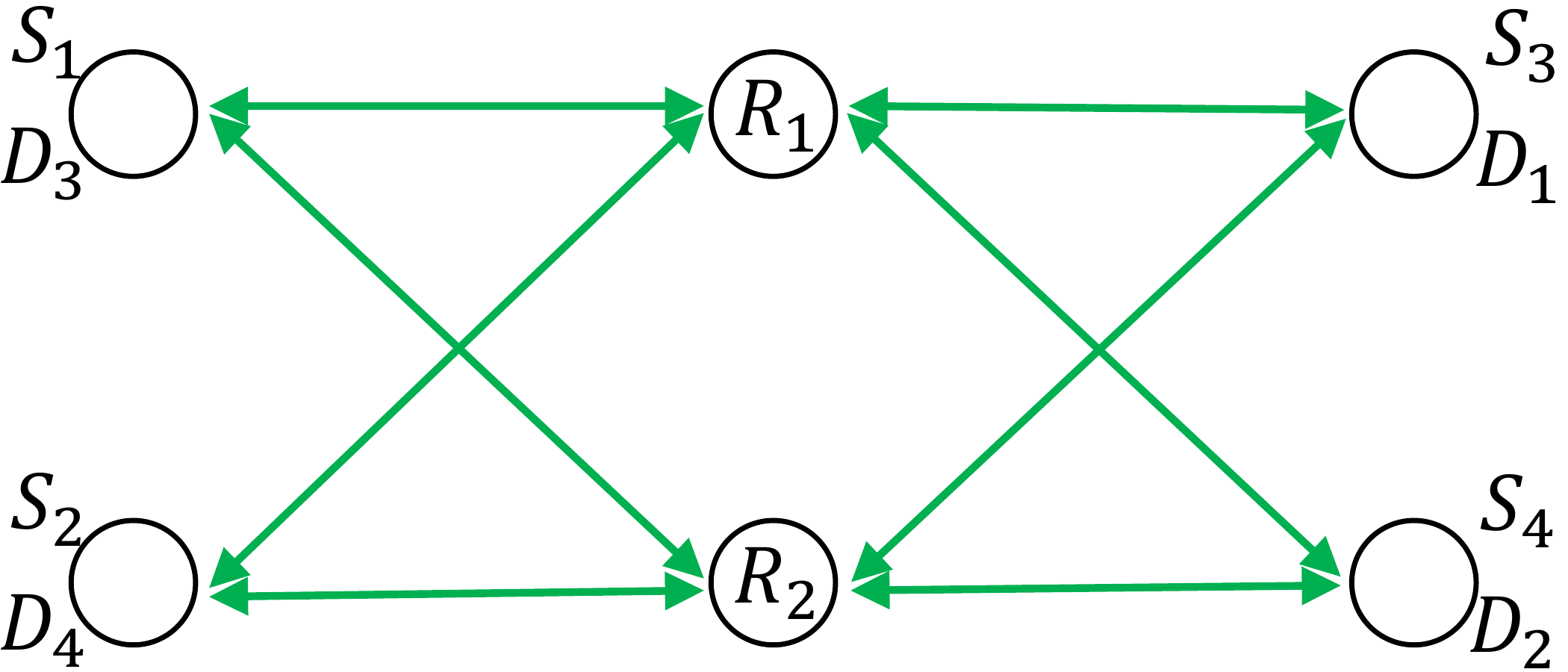}
\caption{Two-way $2\times2\times2$ relay network.}
\label{fig:1}
\end{figure}

\begin{figure}[htbp]
\centering
\subfigure[The channels from transmitters to the relays.]{
	\includegraphics[width=10cm]{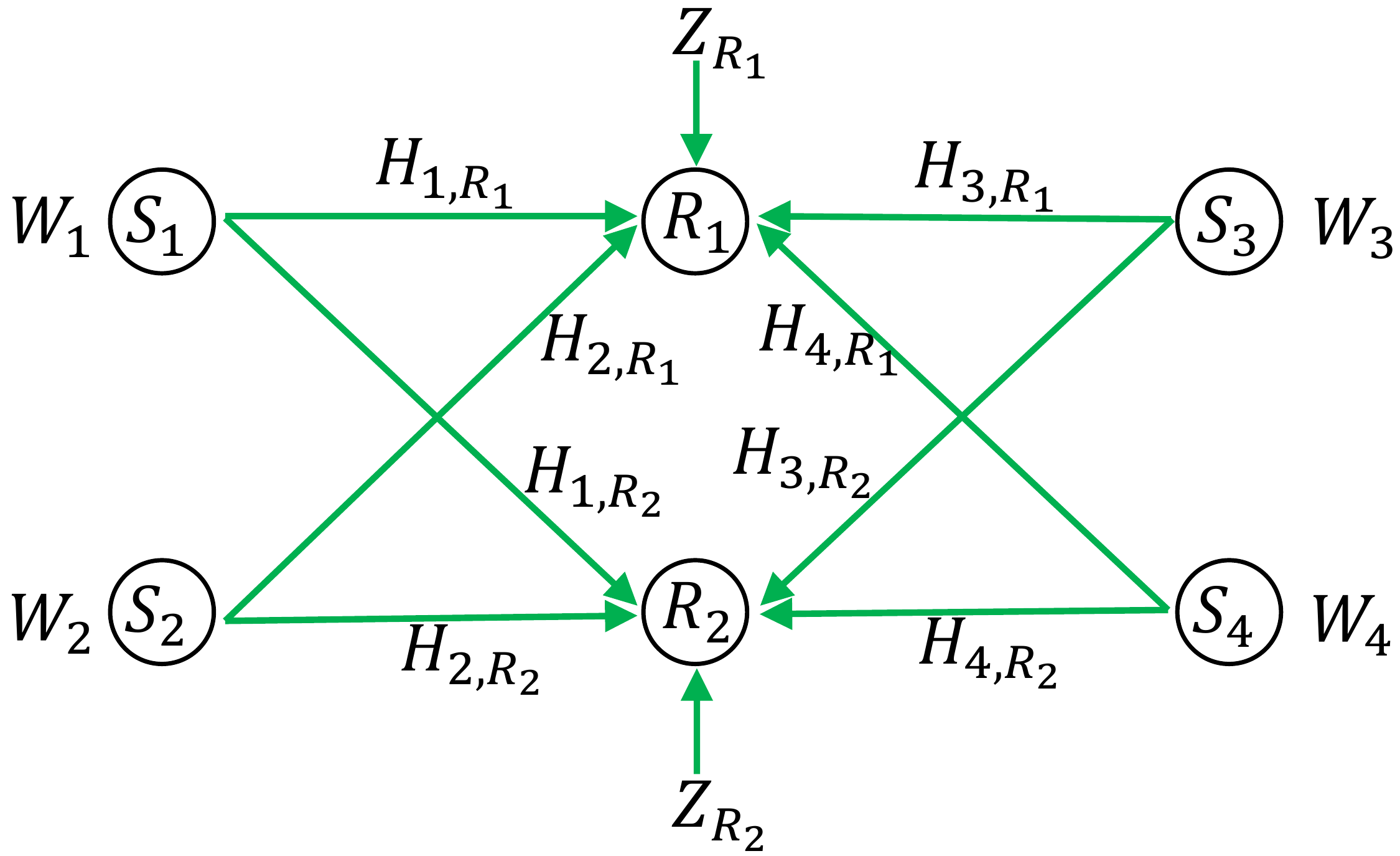}
    \label{fig:2.1}
}
\subfigure[The channels from relays to the receivers.]{
	\includegraphics[width=10cm]{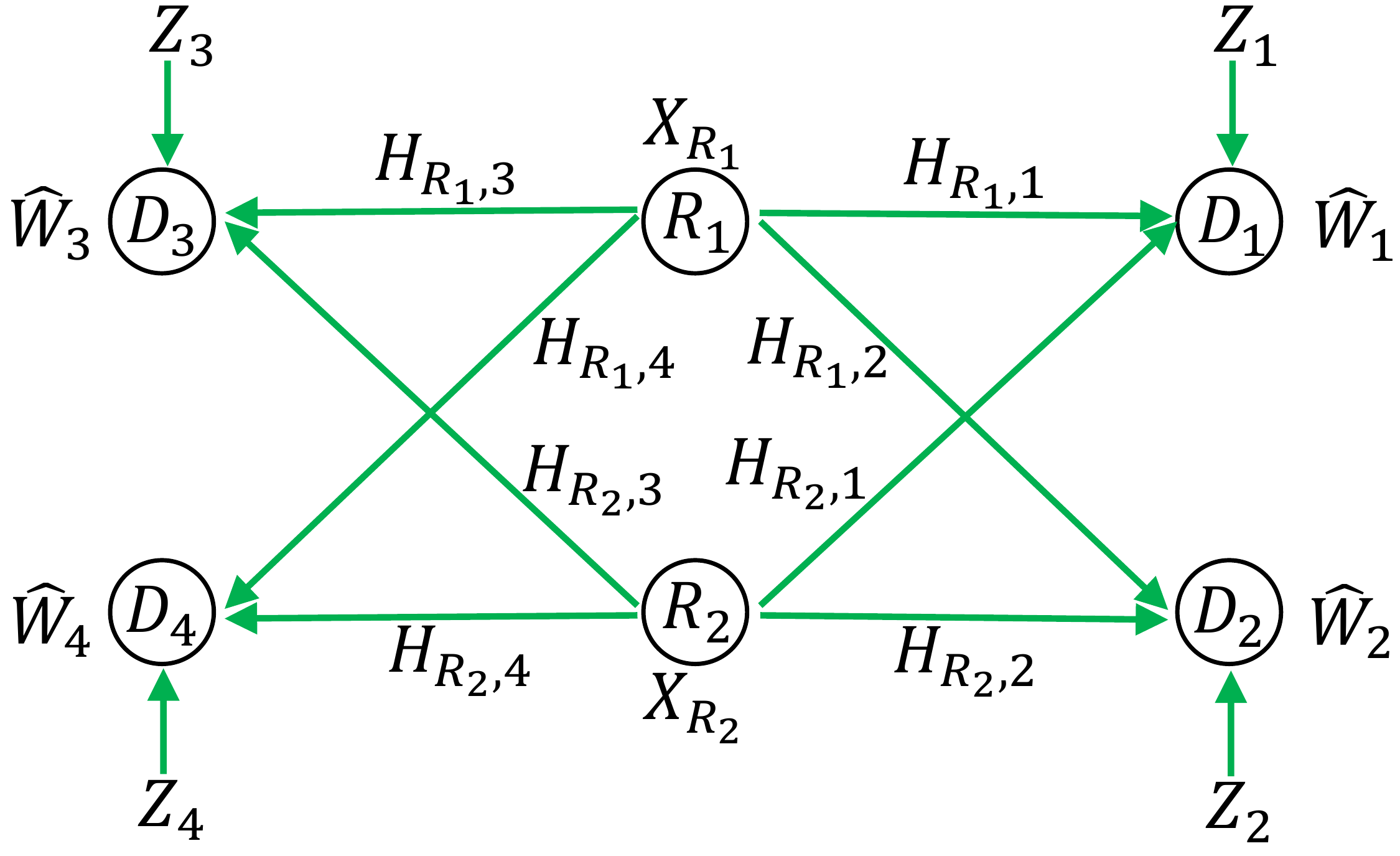}
    \label{fig:2.2}
}
\caption[Optional caption for list of figures]{The channels from and to relays in a two-way $2\times2\times2$ relay network.}
\label{fig:2}
\end{figure}

\subsection{Related Works}

In the literature, there has been extensive research over the last decade to characterize the DoF and the capacity region of one-way relay networks as well as two-unicast networks. However, beyond single-hop, there is much less known about the capacity of multi-flow networks. Even in the simplest case with two sources $S_1$ and $S_2$ and two destinations $D_1$ and $D_2$, there are very few results, such as \cite{hu1963multi}, where the maximum flow in two-unicast undirected wireline networks is characterized. In the wireless realm, constant-gap capacity approximations for specific two-hop networks (the ZZ and ZS structures as depicted in Fig. \ref{fig:rr}) were obtained in \cite{mohajer2009approximate}. Furthermore, it was shown that the network resulting from the concatenation of two fully-connected interference channels (the XX network as depicted in Fig. \ref{fig:222}) admits the maximum of two DoF \cite{shomorony2013two,gou2012aligned}. The achievability scheme relies on the notion of real interference alignment, which was introduced in \cite{motahari2014real}.

In \cite{shomorony2013two}, two-unicast multi-hop wireless networks with two sources $S_1$ and $S_2$ and two destinations $D_1$ and $D_2$ that have a layered structure with arbitrary connectivity are studied. It is shown that, if the channel gains are chosen independently according to continuous distributions, then, with probability $1$, the DoF of the two-unicast layered Gaussian networks can be $1$, $3/2$ or $2$. In particular, for the one-way $2\times2\times2$ relay network in Fig. \ref{fig:222}, one DoF for each user is achievable, i.e., the total DoF is two.

\begin{figure}[htbp]
\centering
\subfigure[One-way ZZ channel.]{
	\includegraphics[width=10cm]{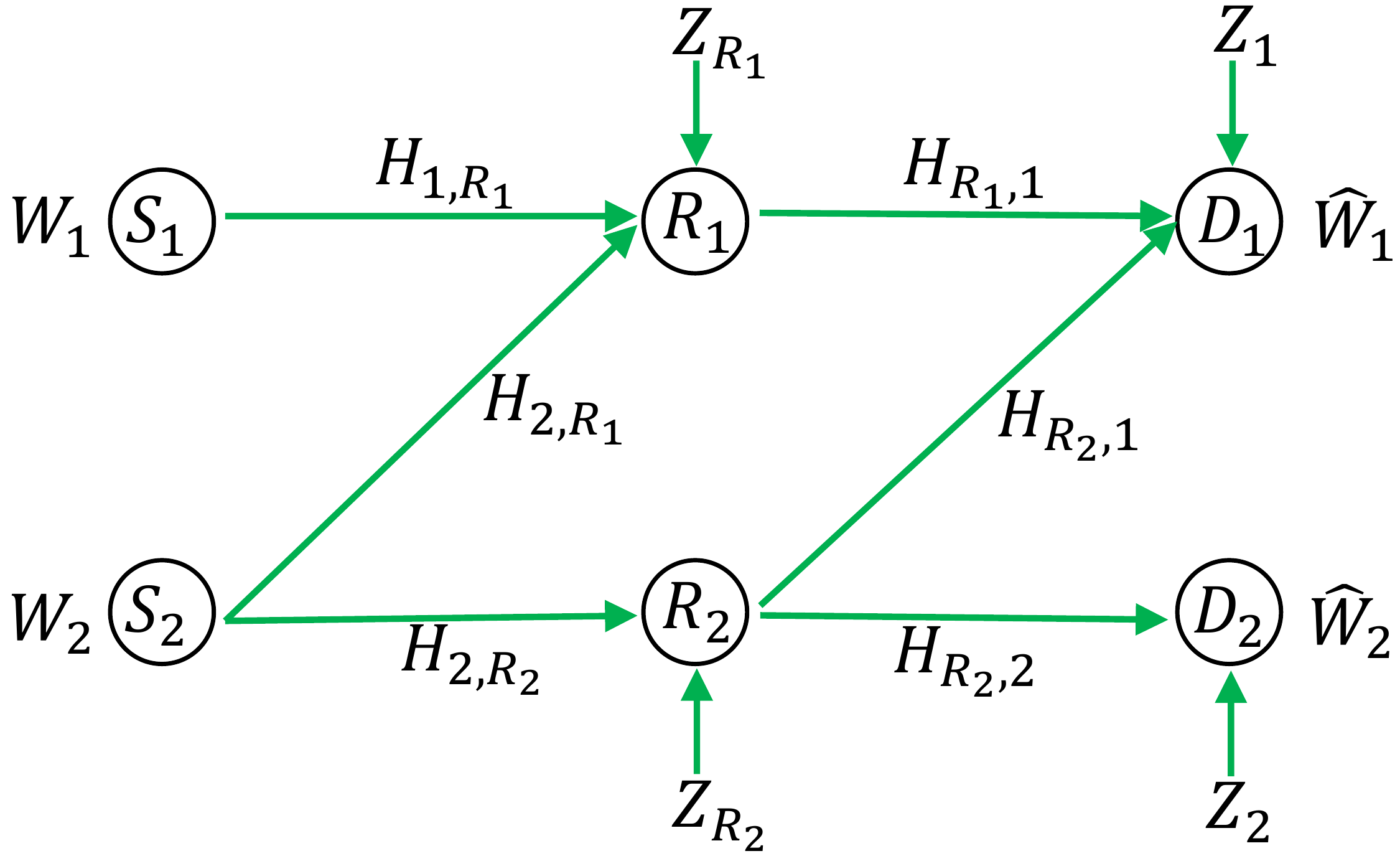}
    \label{fig:2.1rw}
}
\subfigure[One-way ZS channel.]{
	\includegraphics[width=10cm]{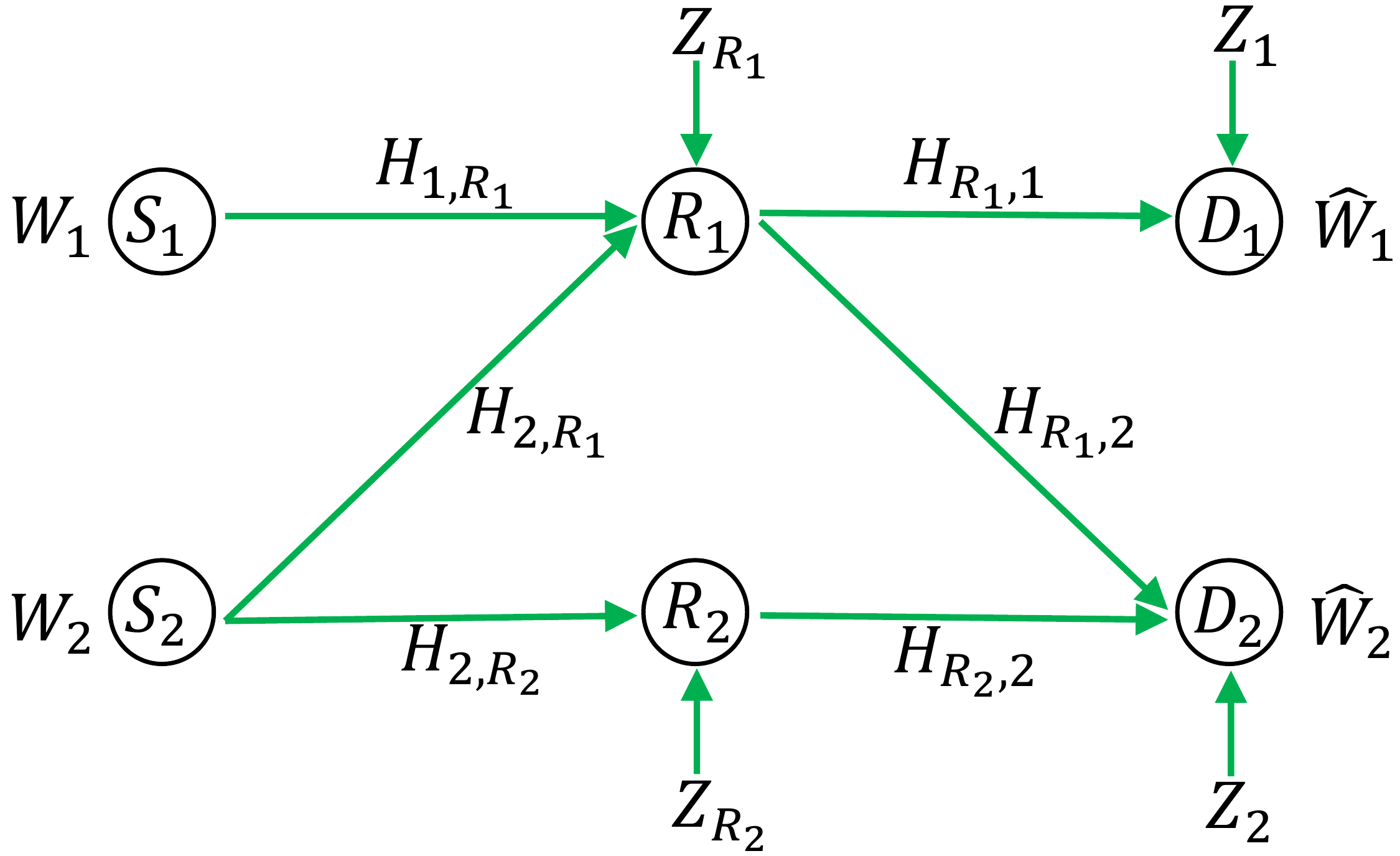}
    \label{fig:2.2rw}
}
\caption[Optional caption for list of figures]{One-way ZZ and ZS networks.}
\label{fig:rr}
\end{figure}

\begin{figure}[htbp]
\centering
	\includegraphics[width=9cm]{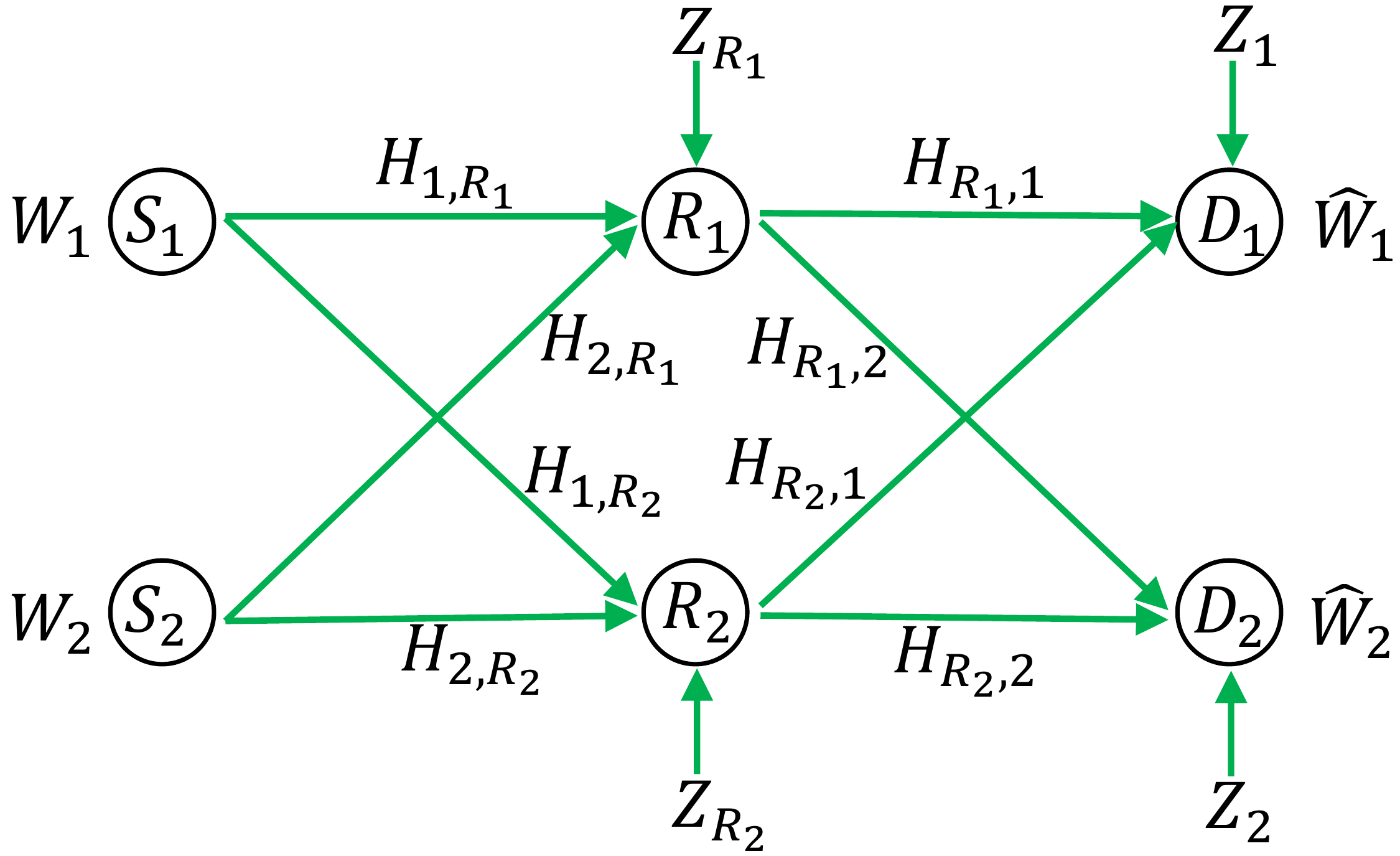}
\caption{One-way $2\times2\times2$ relay network.}
\label{fig:222}
\end{figure}

There are limited number of works on the two-way $2\times2\times2$ relay network in Fig. \ref{fig:1}. In \cite{lee2013achievable}, three different achievability strategies for such a network with MIMO channels are proposed. However, these schemes are considerably away from the optimum, since the achievable total DoF is only two for the SISO case, i.e., the same as the one-way network. In addition, a symmetric finite-field two-way $2\times2\times2$ relay network model is studied in \cite{maier2013cyclic,hong2013two}.

There are a few recent papers that studied the impact of caching on DoF. In particular, \cite{han2015degrees,han2015improving} analyzed the DoF gain induced by caching in interference networks, and proposed a cache-induced cooperative transmission strategy. Also \cite{navid} studied some fundamental limits of the DoF for cache-aided interference networks.

\section{Main Results on the Two-way $2\times2\times2$ Relay Network}\label{thm_2_UB_Sec}

In this section, we study the model in Figs. \ref{fig:1} and \ref{fig:2}, with and without caching at relays.


\subsection{Two-way $2\times2\times2$ Relay Network without Caching}\label{thm_2_UB_Secqq}

We assume that the channel parameters ${H}_{j,R_k}$, ${H}_{R_k,j}$, $j\in\{1,2,3,4\}$,  $k\in\{1,2\}$ are independent and chosen from  the same continuous distribution. Our result is that without caching at the relay, the total DoF of the two-way $2\times 2 \times 2$ network is lower bounded by $2$ and upper bounded by $8/3$.

\begin{proposition}\label{thm_2_LB2}
For a two-way $2\times2\times2$ relay network, ${\text{DoF}}_{NC} \ge 2$.
\end{proposition}
\begin{proof}
If all nodes except for $S_1$, $R_2$, and $S_3$ in Fig. \ref{fig:1} are silent, then the channel can be seen as a two-way $1\times1\times1$ relay network formed by $S_1$, $R_2$, and $S_3$. This channel can achieve two DoF by simply forwarding the sum of the received signals at relay $R_2$, which is the sum of the two messages from $S_1$ and $S_3$. 
\end{proof}
The next result shows that the ${\text{DoF}}$ for the two-way $2\times2\times2$ relay network is upper bounded by $8/3$. Thus, the DoF for the two-way network is smaller than twice the DoF for the one-way network. 

\begin{theorem}\label{thm_2_LB}
For a two-way $2\times2\times2$ relay network, $DoF_{NC} \le \frac{8}{3}$.
\end{theorem}
\begin{proof}
For the outer bound, we assume that there is a channel with infinite capacity between the relays. Also, suppose that a genie provides $W_1$ to this combined relay. Then $W_3$ should be decodable at the relay as it is decodable at $D_3$ given $W_1$. Following this, the messages $W_2$ and $W_4$ can be decoded if the matrix ${\bf H}$ (defined below) is full rank as \eqref{gghhjh} will suggest this mathematically, which happens with probability $1$ over generic channel gains. Therefore, the combined relay should be able to decode three signals $W_2$, $W_3$, and $W_4$ with its two antennas (suggesting $d_2+d_3+d_4\le2$). This is further proved in the following.

Consider $n$ time slots of the channel use and assume that $nR_i$ represents the maximum rate achievable for transmitter $i$ in the total $n$ time slots. Define $Y_{i}^{n}\triangleq(Y_{i}[1],\dots,Y_{i}[n])$ and $X_{i}^{n}\triangleq(X_{i}[1],\dots,X_{i}[n])$, $i=1,\dots,4$. We also define ${\bf h}_{j,R} \triangleq [H_{j,R_1} \ H_{j,R_2}]$,   ${\bf y}_R^n \triangleq [Y_{R_1}^n \ Y_{R_2}^n]$, and ${\bf z}_R^n \triangleq [Z_{R_1}^n \ Z_{R_2}^n]$, where $Y_{R_k}^{n}\triangleq(Y_{R_k}[1],\dots,Y_{R_k}[n])$ and $Z_{R_k}^{n}\triangleq(Z_{R_k}[1],\dots,Z_{R_k}[n])$. Then, we have:
\begin{eqnarray}\label{gghhjh}
nR_3 &\stackrel{(a)}{\le}& I(W_3;Y_3^n)+n\epsilon_n\nonumber\\
&\stackrel{(b)}{\le}& I(W_3;Y_3^n|W_1)+n\epsilon_n\nonumber\\
&\stackrel{(c)}{\le}& I(W_3;{\bf y}_R^n|W_1)+n\epsilon_n\nonumber\\
&=& h({\bf y}_R^n|W_1)-h({\bf y}_R^n|W_1,W_3)+n\epsilon_n\nonumber\\
&=& h({\bf y}_R^n|W_1)-I({\bf y}_R^n;W_2,W_4|W_1,W_3)-h({\bf y}_R^n|W_1,W_3,W_2,W_4)+n\epsilon_n\nonumber\\
&\stackrel{(d)}{=}& h({\bf y}_R^n|W_1)-I({\bf y}_R^n;W_2,W_4|W_1,W_3)-h({\bf z}_R^n)+n\epsilon_n\nonumber\\
&=& h({\bf y}_R^n|W_1)-I({\bf y}_R^n;W_2,W_4|W_1,W_3)-\log\left((2\pi e)^{2n}\right)+n\epsilon_n\nonumber\\
&=& h({\bf y}_R^n|W_1)-H(W_2,W_4|W_1,W_3)+H(W_2,W_4|W_1,W_3,{\bf y}_R^n)-2n\log\left(2\pi e\right)+n\epsilon_n\nonumber\\
&\le& h({\bf y}_R^n|W_1)-H(W_2,W_4|W_1,W_3)+H(W_2,W_4|{\bf y}_R^n-H_{1,R}X_1^n-H_{3,R}X_3^n)-2n\log\left(2\pi e\right)+n\epsilon_n\nonumber\\
&=& h({\bf y}_R^n|W_1)-H(W_2,W_4|W_1,W_3)+H(W_2,W_4|H_{2,R}X_2^n+H_{4,R}X_4^n+{\bf z}_R^n)-2n\log\left(2\pi e\right)+n\epsilon_n\nonumber\\
&\stackrel{(e)}{=}& h({\bf y}_R^n|W_1)-H(W_2,W_4|W_1,W_3)+H(W_2,W_4|[{X}_2^n \ {X}_4^n]+{\bf z}_R^n{\bf H}^{-1})-2n\log\left(2\pi e\right)+n\epsilon_n\nonumber\\
&\stackrel{(f)}{=}& h({\bf y}_R^n|W_1)-H(W_2,W_4|W_1,W_3)+H(W_2,W_4|{X}_2^n+z^{n{'}}_{2},{X}_4^n+z^{n'}_{4})-2n\log\left(2\pi e\right)+n\epsilon_n\nonumber\\
&\stackrel{(g)}{\le}& h({\bf y}_R^n|W_1)-H(W_2,W_4|W_1,W_3)+H(W_2|{X}_2^n+z^{n'}_{2})+H(W_4|{X}_4^n+z^{n'}_{4})-2n\log\left(2\pi e\right)+n\epsilon_n\nonumber\\
&\stackrel{(h)}{=}& h({\bf y}_R^n|W_1)-H(W_2,W_4|W_1,W_3)-2n\log\left(2\pi e\right)+n\epsilon^{'}_n+n\epsilon_n\nonumber\\
&=& h({\bf y}_R^n|W_1)-H(W_2,W_4)-2n\log\left(2\pi e\right)+n\epsilon^{''}_n\nonumber\\
&\stackrel{(i)}{\le}& h({\bf y}_R^n)-H(W_2,W_4)-2n\log\left(2\pi e\right)+n\epsilon^{''}_n\nonumber\\
&\stackrel{(j)}{\le}& h(Y_{R_1}^n)+h(Y_{R_2}^n)-H(W_2,W_4)-2n\log\left(2\pi e\right)+n\epsilon^{''}_n\nonumber\\
&\stackrel{(k)}{\le}& 2\left(\log\left(2\pi e(4H_{\max}^2P+1)\right)^{n}\right)-H(W_2,W_4)-2n\log\left(2\pi e\right)+n\epsilon^{''}_n,
\end{eqnarray}
where $(a)$ follows since the transmission rate is less than or equal to the mutual information between the message and the received signal, and $\epsilon_n$ can be arbitrarily small by increasing $n$; $(b)$ follows since $I(W_3;{\bf Y}_3^n|W_1)-I(W_3;{\bf Y}_3^n)=I(W_3;{\bf Y}_3^n;W_1)\ge-\min\{I(W_3;{\bf Y}_3^n),I(W_1;{\bf Y}_3^n),I(W_3;W_1)\} = 0$ {\large(}as $I(W_3;W_1)=0${\large)}; $(c)$ holds since $W_3\rightarrow {\bf y}_R^n\rightarrow Y_3^n$; $(d)$ follows since by subtracting the contributions of $X_i^n$, $i=1,\dots,4$ from ${\bf y}_R^n$, we will only have Gaussian noise at the relays; $(e)$ follows from the fact that  by defining ${\bf H}\triangleq[{\bf h}^T_{2,R} \ {\bf h}^T_{4,R}]^T$, we obtain the following:
\begin{equation}
\left({\bf y}_R^n-H_{1,R}X_1^n-H_{3,R}X_3^n\right){\bf H}^{-1}=[X_2^n \ X_4^n]^T+{\bf z}_R^n{\bf H}^{-1};
\end{equation}
$(f)$ holds by defining $[z^{n'}_{2} \ z^{n'}_{4}]\triangleq{\bf z}_R^n{\bf H}^{-1}$; $(g)$ holds since conditioning decreases entropy; $(h)$ follows from Fano's inequality and the fact that probability of error in decoding $W_i$ given ${X}_i^n+z^{n'}_{i}$, $i=2,4$ goes to zero for high SNR; $(i)$ holds because conditioning decreases the entropy; $(j)$ holds since $h(X,Y)\le h(X)+h(Y)$; $(k)$ holds since $Y_{R_i}$ is in the form of \eqref{wyek1}, with $|H_{i,R_k}[m]|\le H_{\max}$, and $X_i\sim \mathcal{CN}(0,P)$.




Dividing both sides of \eqref{gghhjh} by $n\log{P}$, and using $n(R_2+R_4-\epsilon^{'''}_n)\le I(W_2;Y_2)+I(W_4;Y_4)=H(W_2)-H(W_2|Y_2)+H(W_4)-H(W_4|Y_4)\le H(W_2)+H(W_4)=H(W_2, W_4)$, results in:
\begin{equation}
\frac{R_3}{\log{P}}{\le} \frac{2\log\left(2\pi e(4H_{\max}^2P+1)\right)^{n}}{n\log{P}}-\frac{(R_2+R_4)}{\log{P}}-\frac{2n\log\left(2\pi e\right)}{n\log{P}}+\frac{\epsilon^{''}_n}{\log{P}},
\end{equation}
and with $n\rightarrow\infty$ and $P\rightarrow\infty$, we obtain the following bound:
\begin{equation}\label{df}
d_2+d_3+d_4 \le 2.
\end{equation}
Similarly, we also have
\begin{eqnarray}
d_1+d_2+d_3 \le 2,\\
d_1+d_2+d_4 \le 2,\\
d_1+d_3+d_4 \le 2.\label{df2}
\end{eqnarray}
Summing up \eqref{df}-\eqref{df2} gives the upper bound in the statement of the theorem.
\end{proof}

\subsection{Symmetric Two-way $2\times2\times2$ Relay Network}\label{sec3}

In this section, we focus on a symmetric case of the two-way $2\times2\times2$ relay network in Fig. \ref{fig:1}, where the channel parameters are assumed to have the following symmetry: ${H}_{1,R_k}={H}_{3,R_k}$, ${H}_{2,R_k}={H}_{4,R_k}$, ${H}_{R_k,1}={H}_{R_k,3}$, and ${H}_{R_k,2}={H}_{R_k,4}$, $k\in\{1,2\}$. The two-hop decomposition of such a symmetric two-way $2\times2\times2$ relay network is shown in Fig. \ref{fig:2m}. The following result shows that the DoF for symmetric two-way $2\times2\times2$ relay network is four. Thus, this two-way network achieves twice the DoF as compared to the one-way network studied in \cite{gou2012aligned}. Essentially the symmetry in channel parameters allows for efficient alignment of the signals from the two directions at the relays that leads to the ${\text{DoF}}$ of four.

\begin{figure}[htbp]
\centering
\subfigure[The channels from transmitters to the relays.]{
	\includegraphics[width=10cm]{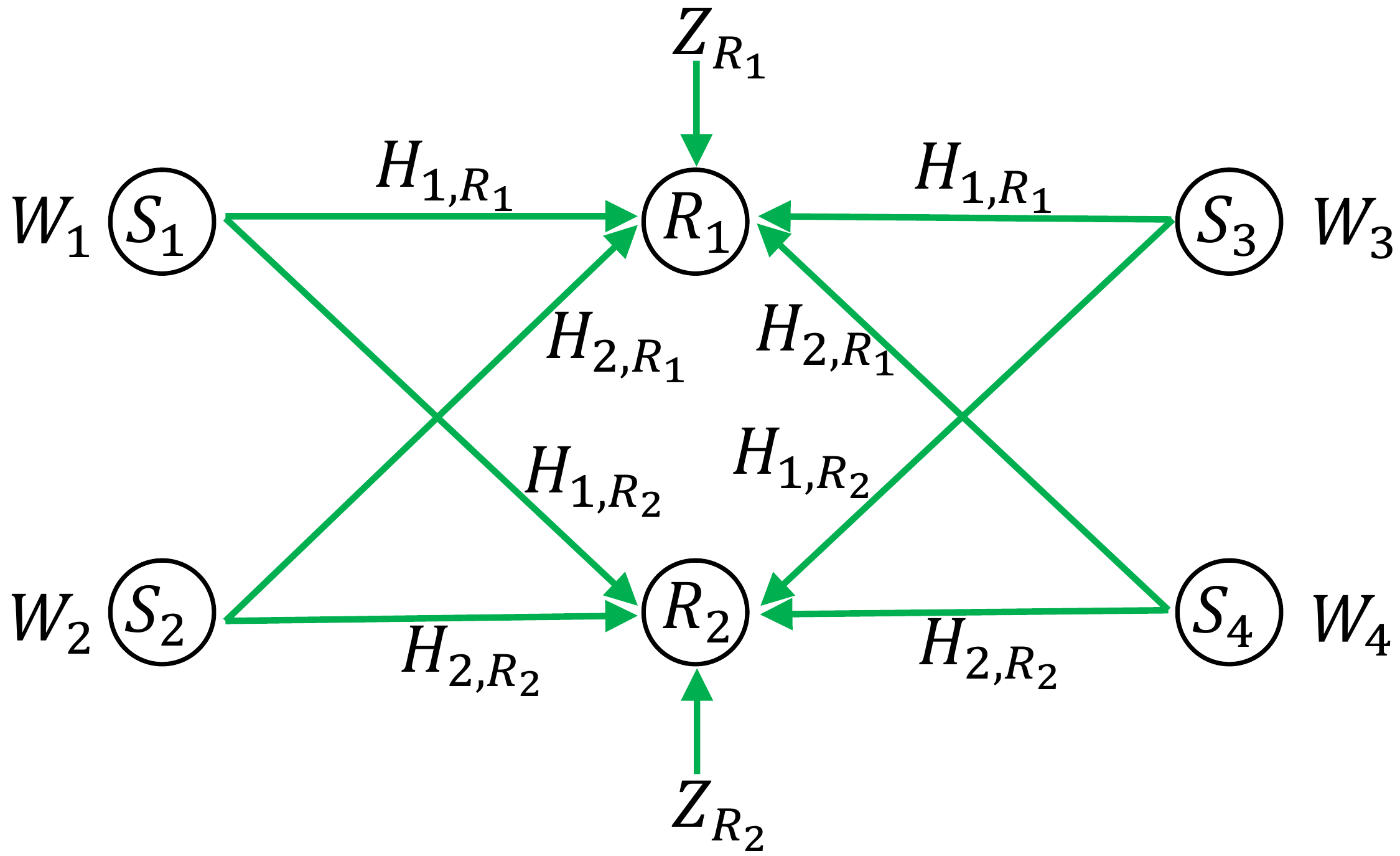}
    \label{fig:2.1m}
}
\subfigure[The channels from relays to the receivers.]{
	\includegraphics[width=10cm]{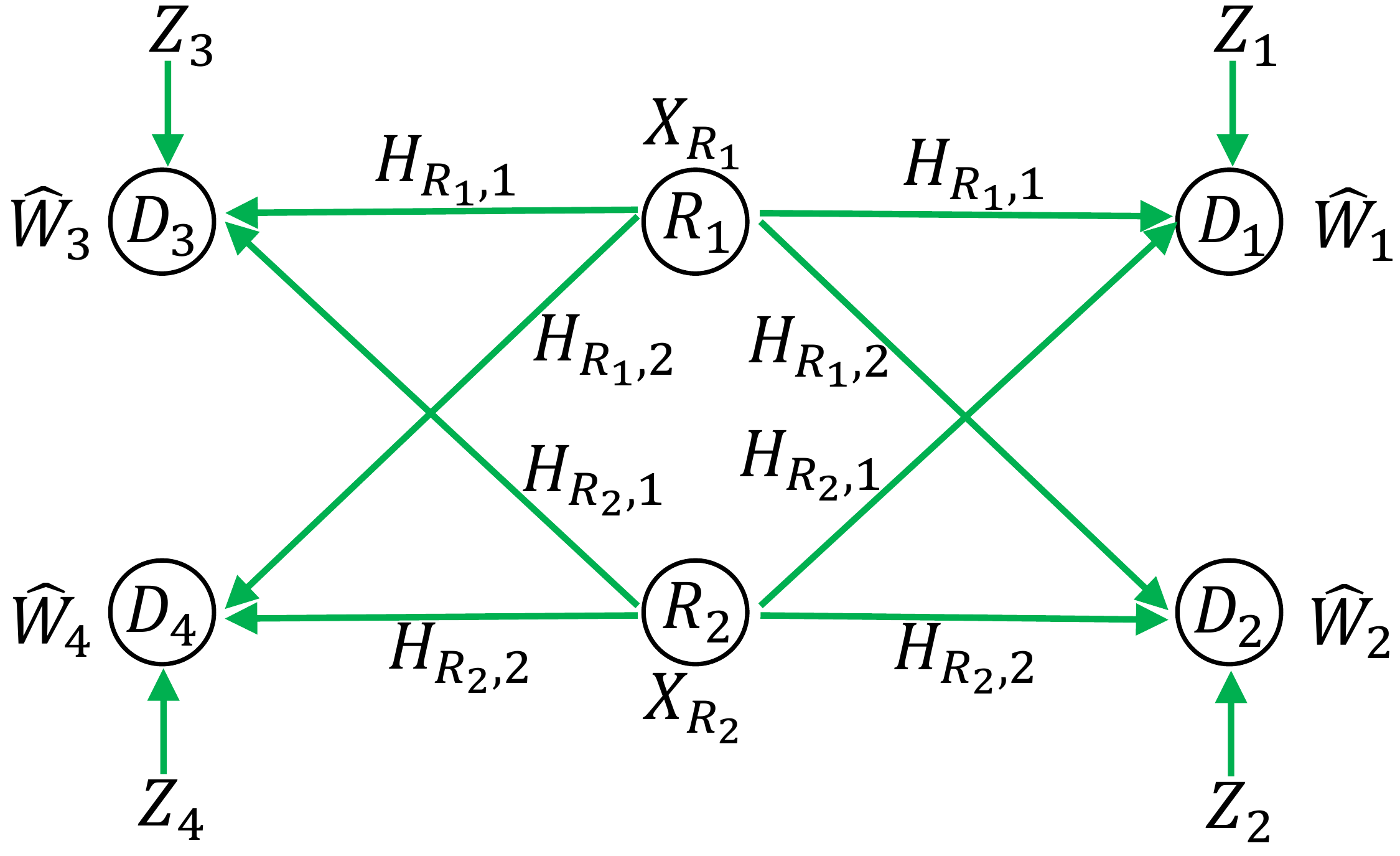}
    \label{fig:2.2m}
}
\caption[Optional caption for list of figures]{The channels from and to relays in a symmetric two-way $2\times2\times2$ relay network.}
\label{fig:2m}
\end{figure}

\begin{proposition}\label{thm_1_LBaa}
For the symmetric two-way $2\times2\times2$ relay network,  $DoF_{NC} = 4$.
\end{proposition}
\begin{proof}
For the achievability, $S_1$ and $S_3$ can be seen as one user and $S_2$ and $S_4$ can be seen as another user from the relay nodes' perspective, due to symmetry. Based on the result  for the one-way $2\times2\times2$ network in \cite{gou2012aligned}, each message can achieve one DoF. Since we increase the number of messages from $2$ in \cite{gou2012aligned} to $4$ in this paper, DoF$=4$ is achievable. The detailed achievability scheme is described in Appendix \ref{Appendix1}.



In addition, it can be seen that the upper bound follows from the cut-set bound.
\end{proof}

\begin{remark}
The comparison of Theorem \ref{thm_2_LB} and Proposition \ref{thm_1_LBaa} shows that, interestingly, even though without caching $8/3$ is an upper bound for generic channels, for symmetric channels $DoF=4$ which shows that this network topology can, in principle, allow $4$ DoF. Note that step $(e)$ in \eqref{gghhjh} for generic channel parameters requires the invertibility of ${\bf H}$, which does not hold for symmetric channels. This is what distinguishes generic channels (for which DoF $\le8/3$) from symmetric symmetric channels (for which DoF $= 4$).
\end{remark}

\section{ Two-way $2\times 2\times 2$ Network with Caching}\label{ofo2}

\subsection{Caching and Transmission Strategy}

In this subsection, we consider the more general model of multi-antenna relays and single-antenna source/destination nodes, where each relay $R_k$, $k\in\{1,2\}$ has $N_k$ antennas. For this model, the difference with the single-antenna case in Section \ref{vmnbx} is that channels $H_{i,R_k}[m]$, $H_{R_k,i}[m]$ are $N_k\times1$ and $1\times N_k$ vectors, respectively, the signals to and from the relays, ($Y_{R_k}[m]$ and $X_{R_k}[m]$, respectively), are vectors of size $N_k\times1$, and the noise $Z_{R_k}[m]$ is an $N_k\times1$ vector, $i\in\{1,2,3,4\}$, $k\in\{1,2\}$. We assume that each relay is equipped with a cache that can store the data from the sources. Our goal is to design strategies for caching and transmission so that the sum rate of all four source-destination pairs is maximized. Similar to caching strategies in the literature \cite{maddah2014fundamental,ji2015throughput}, the transmission consists of two phases. The first phase is the transmission from sources to the relays, as shown in Fig. \ref{fig:2.1}, which is performed offline and is known as the placement phase. The second phase is the transmission from relays to the destinations, as shown in Fig. \ref{fig:2.2}, which is performed online and is known as the delivery phase. We assume that the relays decode ${W}_i$, $i=1,\dots,4$ in the offline phase and save $W'_1\triangleq {W}_1\oplus {W}_3$, $W'_2\triangleq {W}_2\oplus {W}_4$ in their caches.  Then since both relays have access to $W'_1$ and $W'_2$, we can consider them together as an $(N_1+N_2)$-antenna relay, transmitting ${\bf x}^n_R=f(W'_1,W'_2)$, which intends to make $W'_1$ decodable at $D_1$ and $D_3$, and $W'_2$ decodable at $D_2$ and $D_4$ in Fig. \ref{fig:2.2}. 


The next result, after a short review of the compound broadcast channels, shows that the DoF for the two-way $2\times2\times2$ relay network with multiple-antenna relays and single-antenna transceivers is lower bounded by $\frac{4(N_1+N_2)}{N_1+N_2+1}$ under the above caching and transmission strategy.

\subsection{Background on Compound Broadcast Channel}

Here, we briefly introduce the compound broadcast channel and list two lemmas that we need. The Gaussian MISO compound broadcast channel comprises one transmitter with  $N$ antennas and $K$ single-antenna receivers. The transmitter transmits $K$ messages, each intended for a different receiver $i$ whose channel state is chosen from a finite set $\{1,\dots,J_i\}$, $i=1,\dots,K$. In the literature there are several results on the behavior of this channel at high SNR, i.e., the ${\text{DoF}}$. We cite the following two lemmas on the lower and upper bounds of the compound broadcast channel, respectively.

\begin{lemma}\label{dfgl2}{\mbox{\cite{gou2009degreesa,maddah2010degrees}}}
For the compound broadcast channel with $N$ antennas at the transmitter, $K$ single-antenna receivers, and $J_i \ge N$ states at receiver $i$, $i = 1, \dots, K$,  the total DoF of $\frac{NK}{N+K-1}$ is achievable.
\end{lemma}

\begin{lemma}\label{dfgl1} {\mbox{\cite{weingarten2007compound}}}
Consider a compound broadcast channel with $N$ antennas at the transmitter, and $K=2$ single-antenna receivers with $J_1 = 1$, $J_2 = 2$. Then the DoF region is outer bounded by the following region
\begin{eqnarray}
\left(\frac{1}{N}\right)d_{1} +d_{2} &\le& 1,\label{qwe1}\\
d_{1} +\left(\frac{1}{N}\right)d_{2} &\le& 1.\label{qwe2}
\end{eqnarray}
\end{lemma}

\subsection{DoF of Two-way $2\times2\times2$ Relay Network with Caching}

In the following, we provide a result on the achievability of the two-way $2\times2\times2$ relay network with multiple-antenna relays and single-antenna transceivers with caching:

\begin{proposition}\label{prop1}
Under the caching and transmission strategy given above,  ${\text{DoF}} \ge \frac{4(N_1+N_2)}{N_1+N_2+1}$ for the two-way $2\times2\times2$ relay network with multiple-antenna relays and single-antenna transceivers.
\end{proposition}
\begin{proof}
In our transmission strategy, the relays amplify-and-forward the encoded data available in their caches. We treat the two relays together as a super-relay with two antennas that has access to $W'_1$ and $W'_2$. The super-relay intends to  make $W'_1$ decoded at $D_1$ and $D_3$, and $W'_2$ decoded at $D_2$ and $D_4$ since each receiver can decode the desired message by cancelling the contribution of its own message. This becomes equivalent to a compound MISO broadcast channel where message $W'_1$ should be received at both $D_1$ and $D_3$, while message $W'_2$ should be received at both $D_2$ and $D_4$ as depicted in Fig. \ref{fig:1q}. Thus, using Lemma \ref{dfgl2} with  $N=N_1+N_2$ and $K=2$, we obtain the DoF of $\frac{2(N_1+N_2)}{N_1+N_2+1}$ which needs to be multiplied by $2$ due the fact that each of the signals $W'_1$ and $W'_2$  is decoded by two receivers in the original channel. Hence we obtain the DoF of $\frac{4(N_1+N_2)}{N_1+N_2+1}$ with caching.
\end{proof}

\begin{figure}[htbp]
\centering
	\includegraphics[width=6.5cm]{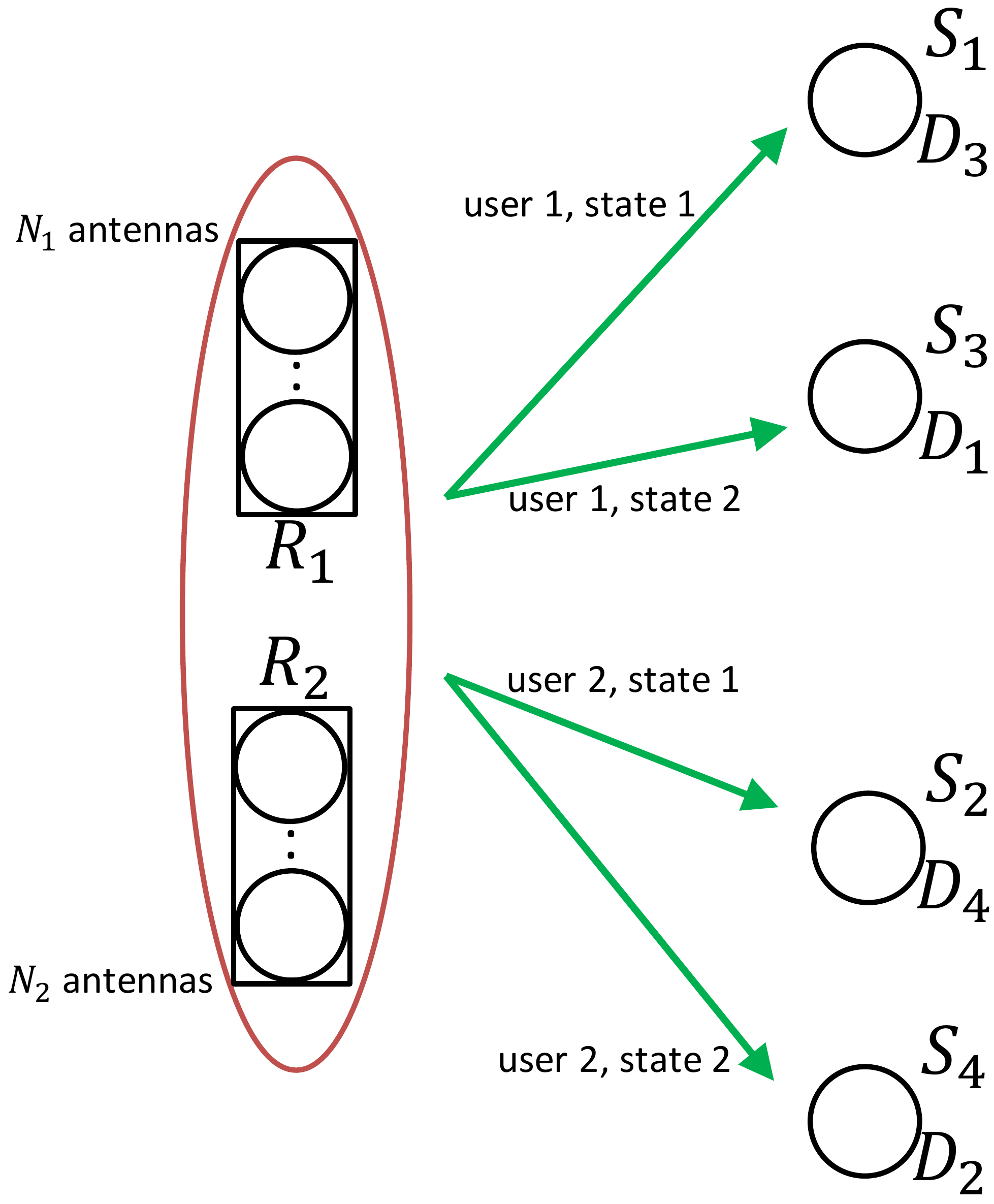}
\caption{ Compound MISO broadcast channel as achievability for two-way $2\times2\times2$ relay network with multiple-antenna relays and single-antenna transceivers and with relay caching.}
\label{fig:1q}
\end{figure}

Now, we provide a discussion on the optimal transmission strategy for the channel with relay caching. Consider the following two transmission strategies for the relays:
\begin{itemize}
\item Encode $W_1$ and transmit  with power $P$: It is helpful for receiver 1, has no effect for receiver 3, and is treated as interference for the other two receivers.
\item Encode $W'_1$ and transmit  with power $P$: It is helpful for receiver 1 (exactly the same effect as in the previous case), and is also helpful for receiver 3 (in contrast to the previous case), and is treated as interference for the other two receivers similar to the previous case.
\end{itemize}
Comparison of the above two strategies suggests that if there is an achievability scheme where all of the messages are decodable for some given total transmission power at the relays, then there is also a strategy with the same power that comprises only $W'_1$ and $W'_2$. Note that receivers 1 and 3 can decode their desired messages by having access to $W'_1$ (because they have access to each other's message and can subtract it), and similar relation holds for receivers 2 and 4.  With this assumption, we present the following result that gives the sum DoF $= \frac{4(N_1+N_2)}{N_1+N_2+1}$.

\begin{proposition}
For the two-way $2\times2\times2$ relay network with multiple-antenna relays and single-antenna transceivers and with caching at the relays, total ${\text{DoF}}_{C} = \frac{4(N_1+N_2)}{N_1+N_2+1}$ if the relays only use $W_1^{'}$ and $W_2^{'}$ in their transmission rather than the original individual messages.
\end{proposition}
\begin{proof}
The achievability follows from Proposition \ref{prop1}. Now we give the proof of the outer bound. With the assumption that the relays only transmit $W'_1$ and $W'_2$. The channel can be seen as a compound broadcast channel with two receivers, where each receiver takes two possible states as in Fig. \ref{fig:1q}. We know that decreasing the number of channel states does not decrease the capacity \cite{gou2009degreesa}. So, we decrease the number of states in receiver 1 to only $1$. Then, according to Lemma \ref{dfgl1}, the DoF region of the compound channel is bounded by \eqref{qwe1}-\eqref{qwe2} for  $N=N_1+N_2$, and bounds $\left(\frac{1}{N}\right)d_{1} +d_{2} \le 1$ and $d_{1} +\left(\frac{1}{N}\right)d_{2}\le1$. These two bounds give a convex region with three non-zero corners of $(d_1,d_2)=\left(\frac{N_1+N_2}{N_1+N_2+1},\frac{N_1+N_2}{N_1+N_2+1}\right)$, $(d_1,d_2)=(1,0)$ and $(d_1,d_2)=(0,1)$. Therefore, $d_1=d_2=\frac{N_1+N_2}{N_1+N_2+1}$ results in the largest value of  $d_1+d_2=\frac{2(N_1+N_2)}{N_1+N_2+1}$ which as before needs to be multiplied by $2$ due the fact that each of the signals $W'_1$ and $W'_2$ is decoded by two receivers in the original channel. Therefore we obtain $\frac{4(N_1+N_2)}{N_1+N_2+1}$ as a DoF upper bound.
\end{proof}
The above result leads to the following conjecture on the general upper bound:
\begin{conjecture}\label{thm_2_UB2w}
For the two-way $2\times2\times2$ relay network with multiple-antenna relays and single-antenna transceivers and with caching at the relays, the total ${\text{DoF}}_{C}\le \frac{4(N_1+N_2)}{N_1+N_2+1}$.
\end{conjecture}

\begin{remark}
The results in Section \ref{thm_2_UB_Secqq} show that the DoF of two-way $2\times2\times2$ relay network with no relay caching is bounded as $2 \le {\text{DoF}}_{NC} \le 8/3$ and Proposition \ref{prop1} shows that ${\text{DoF}}_{C} = 8/3$ is achievable with relay caching for $N_1=N_2=1$. Hence caching can achieve the upper bound of the non-caching DoF of the two-way multiple-unicast network, thus showing that relay caching in this network could potentially improve the DoF.
\end{remark}

Finally, the following corollary gives the DoF of symmetric $2\times2\times2$ relay networks with caching. 


\begin{corollary}
For the symmetric two-way $2\times2\times2$ relay network with relay caching, ${\text{DoF}}_{C}= 4$.
\end{corollary}
\begin{proof}
The proof follows from the following facts together:
\begin{itemize}
\item For the symmetric two-way $2\times2\times2$ relay network with no caching, ${\text{DoF}}_{NC}= 4$ (Proposition \ref{thm_1_LBaa}).
\item For the two-way $2\times2\times2$ relay network, ${\text{DoF}}_{C}\le 4$ due to the cut-set bound.
\item $DoF_{C} \ge DoF_{NC}$.
\end{itemize}
\end{proof}

\section{Conclusions}\label{sec5}
We have investigated the two-way $2\times2\times2$ relay network, a class of four-unicast networks. We have shown that the total DoF is bounded from above by $8/3$, indicating that bidirectional links do not double the DoF. We have also shown that DoF of $8/3$ is achievable with caching at the relays. Therefore, the proposed work demonstrates that caching can achieve the outer bound of the non-caching DoF of the two-way multiple-unicast network, thus showing that relay caching in this network could be helpful in terms of DoF. Finding the DoF for other models of two-way four-unicast networks as well as two-way networks with more than four source-destination pairs remains an open problem. Moreover, the effect of finite-size cache and message popularity in the information theoretic setting also remains to be investigated.


\begin{appendices}
\section{Achievability Strategy for Proposition \ref{thm_1_LBaa} }
\label{Appendix1}

We propose an achievability strategy for the symmetric two-way $2\times2\times2$ relay network by applying the time-extension method \cite{gou2012aligned}. Here, we will see that on $M$-symbol extension of the channel, a ${\text{DoF}}$ of $M$ is achievable for sources $1$ and $3$,{\footnote{In other words, the rate of ${{\mathcal R}_i}=M {\log P} + o({\log P})$, for $i\in\{1,3\}$ is achievable over $M$ time slots, for large $P$.}} while a ${\text{DoF}}$ of $(M-1)$ is achievable for sources $2$ and $4$, thus achieving a total ${\text{DoF}}$ of $(4M-2)$. Therefore, the achievable ${\text{DoF}} = \frac{4M-2}{M}\rightarrow 4$ as $M\rightarrow\infty$.

\noindent
\underline{1. Beamformers at transmitters for signal alignment at relays:} 

With $M$-symbol extension, we consider $M$ consecutive transmissions of the original channel as a single transmission of an equivalent MIMO channel with $M$ antennas and a diagonal channel matrix. In particular, we write the $M$ consecutively received signals at $R_k$ in \eqref{wyek1} in vector form as
\begin{eqnarray}\label{gfhj}
&&\underbrace{\left[ \begin{array}{c}
Y_{R_k}[Mn+1] \\
Y_{R_k}[Mn+2] \\
\vdots \\
Y_{R_k}[Mn+M] \end{array}
\right]}_{{\bf y}_{R_k}[n]} = \underbrace{\left[ \begin{array}{c}
Z_{R_k}[Mn+1] \\
Z_{R_k}[Mn+2] \\
\vdots \\
Z_{R_k}[Mn+M] \end{array}
\right]}_{{\bf z}_{R_k}[n]} +\nonumber\\
&&\sum_{i=1}^{4}
\underbrace{\left[ \begin{array}{cccc}
H_{i,R_k}[Mn+1] & 0 & \cdots & 0 \\
0 & H_{i,R_k}[Mn+2] & \cdots &0 \\
\vdots & \vdots & \ddots & \vdots \\
0 & 0 & \cdots & H_{i,R_k}[Mn+M] \end{array}
\right]}_{{\bf H}_{i,R_k}[n]}
\underbrace{\left[ \begin{array}{c}
X_{i}[Mn+1] \\
X_{i}[Mn+2] \\
\vdots \\
X_{i}[Mn+M] \end{array}
\right]}_{{\bf x}_{i}[n]},
\end{eqnarray}
for $k\in\{1,2\}$. Note that due to the channel symmetry assumption, we have ${\bf H}_{1,R_k}={\bf H}_{3,R_k}$, ${\bf H}_{2,R_k}={\bf H}_{4,R_k}$. For simplicity, we omit the time index $n$ in the following.

We assume that transmitter $S_i$, $i\in\{1,3\}$, has $M$ messages $W_{i,k}$, $k\in\{1,\dots,M\}$ to send. Each message $W_{i,k}$ is encoded into a Gaussian codeword of length $n$ denoted by $x_{{i,k}}[1],\dots,x_{{i,k}}[n]$. $S_i$ transmits the symbol $x_{{i,k}}$ via a beamforming vector ${\bf v}_{{1,k}}$, i.e., its transmitted signal is
\begin{equation}
{\bf x}_i=\sum_{k=1}^{M}{\bf v}_{1,k}x_{i,k},  \ \ \ i\in\{1,3\} .
\end{equation}
Further, we assume that transmitter $S_i$, $i\in\{2,4\}$, has $M-1$ messages $W_{i,k}$, $k\in\{1,\dots,M-1\}$ to send. Each message $W_{i,k}$ is encoded into a Gaussian codeword of length $n$ denoted by $x_{{i,k}}[1],\dots,x_{{i,k}}[n]$. $S_i$ transmits symbol $x_{{i,k}}$ via a beamforming vector ${\bf v}_{{2,k}}$, i.e., its transmitted signal is
\begin{equation}
{\bf x}_i=\sum_{k=1}^{M-1}{\bf v}_{2,k}x_{i,k}, \ \ \ i\in\{2,4\}.
\end{equation}

The beamforming vectors ${\bf v}_{{i,k}}$, $i\in\{1,2\}$ are such that the signals align at relays. In particular, at relay $R_1$, the signals ${\bf H}_{1,R_1}{\bf v}_{{1,i+1}}$, and ${\bf H}_{2,R_1}{\bf v}_{{2,i}}$ are aligned, i.e.,
\begin{equation}\label{allign1}
{\bf H}_{1,R_1}{\bf v}_{{1,i+1}}={\bf H}_{2,R_1}{\bf v}_{{2,i}}, \ \ \ i\in\{1,\cdots,M-1\}.
\end{equation}
At relay $R_2$, the signals ${\bf H}_{1,R_2}{\bf v}_{{1,i}}$, and ${\bf H}_{2,R_2}{\bf v}_{{2,i}}$ are aligned, i.e.,
\begin{equation}\label{allign2}
{\bf H}_{1,R_2}{\bf v}_{{1,i}}={\bf H}_{2,R_2}{\bf v}_{{2,i}}, \ \ \ i\in\{1,\cdots,M-1\}.
\end{equation}
Using \eqref{allign1} and \eqref{allign2}, we can express all beamformers in terms of ${\bf v}_{1,1}$ as follows
\begin{eqnarray}
{\bf v}_{{1,i+1}}&=&{\left({{\bf H}_{1,R_1}^{-1}{\bf H}_{2,R_1}{\bf H}_{2,R_2}^{-1}{\bf H}_{1,R_2}}\right)}^{i}{\bf v}_{{1,1}},\label{allign3}\\
{\bf v}_{{2,i}}&=&{\left({{\bf H}_{2,R_2}^{-1}{\bf H}_{1,R_2}{\bf H}_{1,R_1}^{-1}{\bf H}_{2,R_1}}\right)}^{i-1}{\bf H}_{2,R_2}^{-1}{\bf H}_{1,R_2}{\bf v}_{{1,1}}, \ \ \ i \in \{1,\dots,M-1\}.\label{allign4}
\end{eqnarray}
Note that all matrices in \eqref{allign3}-\eqref{allign4} are diagonal and once ${\bf v}_{1,1}$ is given, all other beamformers can be obtained. Then from \eqref{gfhj},  the received signal at $R_1$ can be written as
\begin{eqnarray}\label{kkq1}
{\bf y}_{R_1}&=&{\bf H}_{1,R_1}{\bf x}_1+{\bf H}_{2,R_1}{\bf x}_2+{\bf H}_{1,R_1}{\bf x}_3+{\bf H}_{2,R_1}{\bf x}_4+{\bf z}_{R_1}\nonumber\\
&=&
{\bf H}_{1,R_1}\sum_{k=1}^{M}{\bf v}_{1,k}x_{1,k}+
{\bf H}_{2,R_1}\sum_{k=1}^{M-1}{\bf v}_{2,k}x_{2,k}+
{\bf H}_{1,R_1}\sum_{k=1}^{M}{\bf v}_{1,k}x_{3,k}+
{\bf H}_{2,R_1}\sum_{k=1}^{M-1}{\bf v}_{2,k}x_{4,k}+
{\bf z}_{R_1}\nonumber\\
&=&{\bf H}_{1,R_1}{\bf v}_{1,1}(x_{1,1}+x_{3,1})+
\sum_{i=1}^{M-1}{\bf H}_{1,R_1}{\bf v}_{1,i+1} \left(x_{1,i+1}+x_{2,i}+x_{3,i+1}+x_{4,i}\right)+{\bf z}_{R_1},
\end{eqnarray}
where the last equality is due to \eqref{allign1}. Similarly the received signal at relay $R_2$ can be written as
\begin{eqnarray}\label{kkq2}
{\bf y}_{R_2}&=&{\bf H}_{1,R_2}{\bf x}_1+{\bf H}_{2,R_2}{\bf x}_2+{\bf H}_{1,R_2}{\bf x}_3+{\bf H}_{2,R_2}{\bf x}_4+{\bf z}_{R_2}\nonumber\\
&=&
{\bf H}_{1,R_2}\sum_{k=1}^{M}{\bf v}_{1,k}x_{1,k}+
{\bf H}_{2,R_2}\sum_{k=1}^{M-1}{\bf v}_{2,k}x_{2,k}+
{\bf H}_{1,R_2}\sum_{k=1}^{M}{\bf v}_{1,k}x_{3,k}+
{\bf H}_{2,R_2}\sum_{k=1}^{M-1}{\bf v}_{2,k}x_{4,k}+
{\bf z}_{R_2}\nonumber\\
&=&
\sum_{i=1}^{M-1}{\bf H}_{1,R_2}{\bf v}_{1,i} \left(x_{1,i}+x_{2,i}+x_{3,i}+x_{4,i}\right)+{\bf H}_{1,R_2}{\bf v}_{1,M}(x_{1,M}+x_{3,M})+{\bf z}_{R_2},
\end{eqnarray}
where the last equality follows from \eqref{allign2}.

If we define ${\bf H}_{R_1}\triangleq[{\bf H}_{1,R_1}{\bf v}_{1,1},{\bf H}_{1,R_1}{\bf v}_{1,2},\dots,{\bf H}_{1,R_1}{\bf v}_{1,M}]$, and ${\bf H}_{R_2}\triangleq[{\bf H}_{1,R_2}{\bf v}_{1,1},{\bf H}_{1,R_2}{\bf v}_{1,2},\dots,{\bf H}_{1,R_2}{\bf v}_{1,M}]$, then from \eqref{kkq1} and \eqref{kkq2} we have
\begin{eqnarray}
{\bf H}_{R_1}^{-1}{\bf y}_{R_1}&=&
\begin{bmatrix}
       x_{1,1}+x_{3,1} \\
       x_{1,2}+x_{3,2}+x_{2,1}+x_{4,1} \\
       \vdots \\
       x_{1,M}+x_{3,M}+x_{2,M-1}+x_{4,M-1}
      \end{bmatrix}
+{\bf H}_{R_1}^{-1}{\bf z}_{R_1}
      \triangleq\begin{bmatrix}
       x_{R_1,1} \\
       \vdots \\
       x_{R_1,M}
      \end{bmatrix},\label{ee1}\\
{\bf H}_{R_2}^{-1}{\bf y}_{R_2}&=&
\begin{bmatrix}
       x_{1,1}+x_{3,1}+x_{2,1}+x_{4,1} \\
       \vdots \\
       x_{1,M-1}+x_{3,M-1}+x_{2,M-1}+x_{4,M-1} \\
       x_{1,M}+x_{3,M}
      \end{bmatrix}
+{\bf H}_{R_2}^{-1}{\bf z}_{R_2}
\triangleq\begin{bmatrix}
       x_{R_2,1} \\
       \vdots \\
       x_{R_2,M}
      \end{bmatrix}.\label{ee2}
\end{eqnarray}

\noindent
\underline{2. Beamforming at relays for interference cancellation at destinations:} 

At time $n$, each relay amplify-and-forwards its received signal at time $(n-1)$. Specifically, $R_1$ sends $x_{R_1,k}$ in \eqref{ee1} using a beamforming vector ${\bf v}_{R_1,k}$, $k\in\{1,\dots,M\}$, i.e.,
\begin{equation}
{\bf x}_{R_1}=\sum_{k=1}^{M}{\bf v}_{R_1,k}x_{R_1,k}.
\end{equation}
$R_2$ sends $x_{R_2,k}$ in \eqref{ee2} using a beamforming vector ${\bf v}_{R_2,k}$, $k\in\{1,\dots,M-1\}$, i.e.,
\begin{equation}
{\bf x}_{R_2}=\sum_{k=1}^{M-1}{\bf v}_{R_2,k}x_{R_2,k}.
\end{equation}
The relay beamformers ${\bf v}_{R_1, k},  {\bf v}_{R_2, k}$ are designed such that the interference is canceled in the signal received at the destinations $D_i$, $i\in\{1,2,3,4\}$. Specifically, as before, with $M$-symbol extension the equivalent MIMO model for the received signal at $D_i$ can be written as
\begin{eqnarray}\label{akheri}
&&\underbrace{\left[ \begin{array}{c}
Y_{i}[Mn+1] \\
Y_{i}[Mn+2] \\
\vdots \\
Y_{i}[Mn+M] \end{array}
\right]}_{{\bf y}_{i}[n]} = \underbrace{\left[ \begin{array}{c}
Z_{i}[Mn+1] \\
Z_{i}[Mn+2] \\
\vdots \\
Z_{i}[Mn+M] \end{array}
\right]}_{{\bf z}_{i}[n]} +\nonumber\\
&&\sum_{k=1}^{2}
\underbrace{\left[ \begin{array}{cccc}
H_{R_k,i}[Mn+1] & 0 & \cdots & 0 \\
0 & H_{R_k,i}[Mn+2] & \cdots &0 \\
\vdots & \vdots & \ddots & \vdots \\
0 & 0 & \cdots & H_{R_k,i}[Mn+M] \end{array}
\right]}_{{\bf H}_{R_k,i}[n]}
\underbrace{\left[ \begin{array}{c}
X_{R_k}[Mn+1] \\
X_{R_k}[Mn+2] \\
\vdots \\
X_{R_k}[Mn+M] \end{array}
\right]}_{{\bf x}_{R_k}[n]},
\end{eqnarray}
for $i\in\{1,\dots,4\}$. Note that due to the channel symmetry assumption, we have ${\bf H}_{R_k,1}={\bf H}_{R_k,3}$, ${\bf H}_{R_k,2}={\bf H}_{R_k,4}$.

From \eqref{ee1}-\eqref{ee2}, ${x}_{R_1, i+1}  = x_{1,i+1}+x_{3,i+1}+x_{2,i}+x_{4,i}$ and ${x}_{R_2, i} = x_{1,i}+x_{3,i}+x_{2,i}+x_{4,i}$ (ignoring noise). In order to cancel $x_{2,i}$ and $x_{4,i}$ at $D_1$ and $D_3$, we choose the relay beamformers such that
\begin{equation}\label{avali}
{\bf H}_{R_1,1}{\bf v}_{R_1,i+1}=
-{\bf H}_{R_2,1}{\bf v}_{R_2,i}, \ \ \  i\in\{1,\dots,M-1\}.
\end{equation}
Moreover, we also have ${x}_{R_1, i} = x_{1,i}+x_{3,i}+x_{2,i-1}+x_{4,i-1}$ and ${x}_{R_2, i} = x_{1,i}+x_{3,i}+x_{2,i}+x_{4,i}$. In order to cancel $x_{1,i}$ and $x_{3,i}$ at $D_2$ and $D_4$, we choose the relay beamformers such that
\begin{equation}\label{dovomi}
-{\bf H}_{R_1,2}{\bf v}_{R_1,i}=
{\bf H}_{R_2,2}{\bf v}_{R_2,i}, \ \ \ \  i\in\{1,\dots,M-1\}.
\end{equation}
Using \eqref{avali} and \eqref{dovomi}, the relay beamformers can be expressed in terms of ${\bf v}_{R_1,1}$ as follows:
\begin{eqnarray}
{\bf v}_{R_1,i+1}&=&{({\bf H}_{R_1,1}^{-1}{\bf H}_{R_2,1}{\bf H}_{R_2,2}^{-1}{\bf H}_{R_1,2})}^{i}{\bf v}_{R_1,1},\label{q1}\\
{\bf v}_{R_2,i}&=&-{({\bf H}_{R_2,2}^{-1}{\bf H}_{R_1,2}{\bf H}_{R_1,1}^{-1}{\bf H}_{R_2,1})}^{i-1}
{\bf H}_{R_2,2}^{-1}{\bf H}_{R_1,2}{\bf v}_{R_1,1}.\label{q2}
\end{eqnarray}
Hence given ${\bf v}_{R_1,1}$, all relay beamformers can be obtained.

\noindent
\underline{3. Successive decoding at destinations:} 

Now we examine the received signals at the destinations. From \eqref{akheri},  the received signal at $D_i$, $i\in\{1,3\}$, can be written as
\begin{eqnarray}\label{qwer1}
{\bf y}_i&=& {\bf H}_{R_1,1}{\bf x}_{R_1}+{\bf H}_{R_2,1}{\bf x}_{R_2}+{\bf z}_{i}\nonumber\\
&=& {\bf H}_{R_1,1}\sum_{k=1}^{M}{\bf v}_{R_1,k}x_{R_1,k}+{\bf H}_{R_2,1}\sum_{k=1}^{M-1}{\bf v}_{R_2,k}x_{R_2,k}+{\bf z}_{i}\nonumber\\
&=& {\bf H}_{R_1,1}{\bf v}_{R_1,1}(x_{1,1}+x_{3,1}+z_{1,1}^{'})+\nonumber\\
&&\sum_{k=1}^{M-1}{\bf H}_{R_1,1}{\bf v}_{R_1,k+1}\left(x_{1,k+1}+x_{3,k+1}-x_{1,k}-x_{3,k}+z_{1,k+1}^{'}-z_{2,k}^{'}\right)
+{\bf z}_{i},
\end{eqnarray}
where the last equality follows from \eqref{avali} and $z_{i,k}^{'}$ is the $k^{\text{th}}$ entry of the noise vector ${\bf H}_{R_i}^{-1}{\bf z}_{R_i}$ in \eqref{ee1} and \eqref{ee2}. Similarly, the received signal at $D_i$, $i\in\{2,4\}$, can be written as
\begin{eqnarray}\label{qwer2}
{\bf y}_i&=& {\bf H}_{R_1,1}{\bf x}_{R_1}+{\bf H}_{R_2,1}{\bf x}_{R_2}+{\bf z}_{i}\nonumber\\
&=& {\bf H}_{R_1,1}\sum_{k=1}^{M}{\bf v}_{R_1,k}x_{R_1,k}+{\bf H}_{R_2,1}\sum_{k=1}^{M-1}{\bf v}_{R_2,k}x_{R_2,k}+{\bf z}_{i}\nonumber\\
&=& {\bf H}_{R_1,1}{\bf v}_{R_1,M}(x_{1,M}+x_{3,M}+x_{2,M-1}+x_{4,M-1}+z_{1,M}^{'})+\nonumber\\
&&\sum_{k=1}^{M-1}{\bf H}_{R_2,1}{\bf v}_{R_2,k}\left(x_{2,k}+x_{4,k}-x_{2,k-1}-x_{4,k-1}+z_{2,k}^{'}-z_{1,k}^{'}\right)
+{\bf z}_{i},
\end{eqnarray}
where the last equality follows from \eqref{dovomi} and $x_{2,0}=x_{4,0}=0$.

If we define ${\bf H}_{1}\triangleq[{\bf H}_{R_1,1}{\bf v}_{R_1,1},\dots,{\bf H}_{R_1,1}{\bf v}_{R_1,M}]$, and ${\bf H}_{2}\triangleq[{\bf H}_{R_2,1}{\bf v}_{R_2,1},\dots,{\bf H}_{R_2,1}{\bf v}_{R_2,M-1},{\bf H}_{R_1,1}{\bf v}_{R_1,M}]$, then from \eqref{qwer1} and \eqref{qwer2} we have
\begin{eqnarray}
{\bf H}_{1}^{-1}{\bf y}_{i}&=&
\begin{bmatrix}
       x_{1,1}+x_{3,1}+z_{1,1}^{'} \\
       x_{1,2}+x_{3,2}-x_{1,1}-x_{3,1}+z_{1,2}^{'}-z_{2,1}^{'} \\
       \vdots \\
       x_{1,M}+x_{3,M}-x_{1,M-1}-x_{3,M-1}+z_{1,M}^{'}-z_{2,M-1}^{'}
      \end{bmatrix}
+{\bf H}_{1}^{-1}{\bf z}_{i}, \ \ \ i\in\{1,3\},\label{eekk1}\\
{\bf H}_{2}^{-1}{\bf y}_{i}&=&
\begin{bmatrix}
       x_{2,1}+x_{4,1}-x_{2,0}-x_{4,0}+z_{2,1}^{'}-z_{1,1}^{'} \\
       \vdots \\
       x_{2,M-1}+x_{4,M-1}-x_{2,M-2}-x_{4,M-2}+z_{2,M-1}^{'}-z_{1,M-1}^{'} \\
       x_{1,M}+x_{3,M}+x_{2,M-1}+x_{4,M-1}+z_{1,M}^{'}
      \end{bmatrix}
+{\bf H}_{2}^{-1}{\bf z}_{i}, \ \ \ i\in\{2,4\}.\label{eekk2}
\end{eqnarray}

Now from \eqref{eekk1}, $D_i$, $i\in\{1,3\}$, is able to first estimate $x_{1,1}+x_{3,1}$ and then add it to the second dimension in order to estimate $x_{1,2}+x_{3,2}$ and so on. Finally, since $D_1=S_3$, $D_1$ knows $x_{3,1},\dots,x_{3,M}$, and therefore it can decode $x_{1,1},\dots,x_{1,M}$. Similarly, since $D_3=S_1$, $D_3$ knows $x_{1,1},\dots,x_{1,M}$, and so it can decode $x_{3,1},\dots,x_{3,M}$.

Similarly from \eqref{eekk2}, $D_i$, $i\in\{2,4\}$, can decode $x_{i,1},\dots,x_{i,M}$. And this completes the proof of achievability.

\end{appendices}

\bibliographystyle{IEEETran}
\bibliography{bib}

\end{document}